\documentclass[oneside,english]{amsart}
\usepackage[T1]{fontenc}
\usepackage[latin9]{inputenc}
\usepackage{graphicx}
\usepackage{amssymb}
\usepackage{hyperref}

\makeatletter
  \theoremstyle{plain}
  \newtheorem{lem}{Lemma}
  \theoremstyle{plain}
  \newtheorem{thm}{Theorem}
\newenvironment{lyxlist}[1]
{\begin{list}{}
{\settowidth{\labelwidth}{#1}
 \setlength{\leftmargin}{\labelwidth}
 \addtolength{\leftmargin}{\labelsep}
 }}
{\end{list}}
  \theoremstyle{remark}
  \newtheorem*{claim*}{Claim}
  \theoremstyle{remark}
  \newtheorem*{acknowledgement*}{Acknowledgment}

\def\blfootnote{\xdef\@thefnmark{}\@footnotetext} 

\usepackage{fancyhdr}
\makeatother

\usepackage{babel}

\begin{document}

\title{Markovian Network Interdiction and \\
the Four Color Theorem}

        

\author{Alexander Gutfraind$^{1}$}
\author{Kiyan Ahmadizadeh$^{2}$}


\keywords{Network Interdiction, Markov Random Walk, Unreactive Evader, Four Color Theorem, Computational Complexity, NP Hard}

\maketitle

\setcounter{footnote}{0}

\begin{abstract}
The Unreactive Markovian Evader Interdiction Problem (UME) asks to
optimally place sensors on a network to detect Markovian motion by
one or more {}``evaders''. It was previously proved that finding
the optimal sensor placement is NP-hard if the number of evaders is
unbounded. Here we show that the problem is NP-hard with just $2$
evaders using a connection to coloring of planar graphs. The results
suggest that approximation algorithms are needed even in applications
where the number of evaders is small. It remains an open problem to
determine the complexity of the 1-evader case or to devise efficient
algorithms.
\end{abstract}

\section{Introduction}
\blfootnote{
$^{1}$ Theoretical Division, Los Alamos National Laboratory, Los Alamos, New Mexico USA 87545, 
\href{mailto:ag362@cornell.edu}{ag362@cornell.edu} 
$^{2}$ Department of Computer Science, Cornell University, Ithaca, New York USA 14853,
\href{mailto:kiyan@cs.cornell.edu}{kiyan@cs.cornell.edu}.}

Network interdiction is a class of discrete optimization problems
originating in applications such as supply chains, sensing and disease
control \cite{Corley-1982-most,Pan-2003-models,Pourbohloul05}. In
network interdiction one or several {}``evaders'' traverse the network
and the objective is to place devices for sensing or capturing the
evaders (the text uses the words {}``sensing'', {}``capturing'' and 
{}``detecting'' interchangeably). The problem is hard in part
because the motion of the evaders is to some extent stochastic. Depending
on the application such stochasticity may be caused by random errors,
systematical misestimation of the network topology, deliberate misdirection
or computational power that is insufficient for path optimization.
The simplest model of this stochasticity is based on a Markov chain
on the nodes of the network \cite{Gutfraind08markovian}.

For a concrete example, consider the problem of placing police units
(the sensors) on the highway network to catch a bank robber (the evader)
moving towards a safehouse. Because of his haste and lack of information
his motion is not predictable with certainty. Another application
is found in problems like electronic network monitoring, where a limited
number of devices must be placed to scan as much of the traffic as
possible even though the traffic moves stochastically.

These and similar applications suggest a formulation of network interdiction
where evaders are (1) Markovian and (2) cannot or do not change their
motion based on where the sensors were placed. The resulting optimization
problem has been termed {}``Unreactive Markovian Evader Interdiction''
(UME) \cite{Gutfraind09unreactive} (see also earlier work in \cite{Berman95}).
In UME the objective is to place sensors at nodes or edges of the
network (generally a directed weighted graph $G(V,E)$), subject to
a cardinality constraint, so as to maximize the probability that the
evader passes through a sensor on his way to his target.

It was shown in \cite{Gutfraind09unreactive} that when the number
of evaders can be arbitrarily large then the UME problem is NP-hard,
but the complexity of UME is an open problem when the number of evaders
is bounded. Such complexity must arise from network structure and
stochasticity of motion, and this problem is addressed here.

An instance of UME contains an evader - a Markov chain given by initial
{}``source'' distribution $\mathbf{a}$ and transition probability
matrix $\mathbf{M}$ with the property that a specified {}``target''
node $t$ is a {}``killing'' state: upon reaching $t$ the evader
is removed from the network. The sensors are represented by a matrix
of decision variables $\mathbf{r}$: $r_{ij}=1$ if $(i,j)$ is interdicted
and $0$ otherwise. If an evader passes through a sensor at edge $(i,j)$
he is detected with probability $d_{ij}\in[0,1]$, termed {}``interdiction
efficiency'' (in general $d_{ij}\neq d_{ji}$ even if the graph is
undirected). The objective is to choose $\mathbf{r}$, subject to
a budget constraint (cardinality of non-zero entries:$\left\Vert \mathbf{r}\right\Vert \leq\beta$),
so as to maximize the probability $J$ of {}``catching'' the evader
before he reaches the target $t$. Under certain restrictions on the
Markov chain (e.g. $t$ is an absorbing state) this probability $J$
can be expressed in closed form \cite{Gutfraind09unreactive}:

\begin{equation}
J({\bf \mathbf{a}},\mathbf{M},\mathbf{r},\mathbf{d})=1-\left({\bf \mathbf{a}}\left[\mathbf{I}-\left(\mathbf{M}-\mathbf{M}\odot{\bf \mathbf{r}}\odot\mathbf{d}\right)\right]^{-1}\right)_{t}\,,\label{eq:evader-cost}\end{equation}
where the symbol $\odot$ means element-wise (Hadamard) multiplication. 

It would be convenient later to use a closely-related problem of {}``node
interdiction'', where the interdiction set is chosen from nodes rather
than edges. Define {}``interdiction of node $i$'' to mean setting
$r_{ij}=1$ $\forall(i,j)\in E$ (that is, interdicting all evaders
leaving $i$). The UME problem on nodes is then the problem of finding
an interdiction set $Q\subset V$ maximizing $J$ subject to $\left\Vert Q\right\Vert \leq B$.

UME can be generalized for applications where there are multiple evaders
or scenarios each realized with probability $w^{(k)}$ ($\sum_{k}w^{(k)}=1$).
Evader $k$ follows a Markov chain $\mathbf{a}^{(k)},\mathbf{M}^{(k)}$
and has probability of capture $J^{(k)}$ found from Eq.~\ref{eq:evader-cost}:
\begin{equation}
J^{(k)}({\bf \mathbf{a}}^{(k)},\mathbf{M}^{(k)},\mathbf{r},{\bf \mathbf{d}})=1-\left({\bf \mathbf{a}}^{(k)}\left[\mathbf{I}-\left(\mathbf{M}^{(k)}-\mathbf{M}^{(k)}\odot\mathbf{r}\odot\mathbf{d}\right)\right]^{-1}\right)_{t^{(k)}}\,.\label{eq:multiple-evader-cost}\end{equation}

The UME objective becomes maximizing the expected probability of capture:
\begin{equation}
\left\langle J\right\rangle =\sum_{k}w^{(k)}J^{(k)}\,.\label{eq:weighted-cost}\end{equation}

The motivation to determine the complexity of the problem is both
theoretical and practical. On the theoretical level it is known that
other formulations of network interdiction (such as where the evaders
react to the interdiction decisions - can see and possibly avoid the
sensors) are NP-hard and hard to approximate \cite{Bar-noy-1995-complexity,Boros06-inapproximability}.
 On the practical level one wishes to explain the difficulty solving
exactly even small instances of UME. Computational experiments described
in \cite{Gutfraind09unreactive} indicate that state-of-the-art integer
programming packages such as CPLEX version $10.1$ may fail to efficiency
solve instances of UME involving $4$ evaders on networks with just
$100$ nodes (runtime $>9$ hours). The proof in the next section
indicates that already with 2-evaders UME is NP-hard, and hence in
general UME can only be attacked using approximation algorithms (unless
P=NP).

\section{UME with $2$ evaders is NP-hard}

The following proof is a reduction of Planar Vertex Cover - an NP-complete
problem \cite{Garey74} (the word {}``vertex'' is used interchangeably
with the word {}``node''). Planar Vertex Cover asks to determine
whether given an undirected planar graph $G'(V',E')$ there exists
a set $C$ of $B'\geq0$ vertices that can {}``cover'' all the edges
of $G'$. The set $C\subset V$ is called a {}``vertex cover'' if
all the edges are incident to at least one vertex in $C$. The proof
constructs an instance of UME with possibly multiple evaders so that
each one of the edges in $E'$ is traversed by at least one evader,
who then immediately moves to a special target node. These conditions
mean that to achieve interdiction with expected probability $=1$
it would be necessary and sufficient to interdict at least one of
the incident vertices of every edge - creating a cover.

The Markovian property of the evaders sets a lower bound on the number
of evaders needed to meet those conditions - a lower bound directly
related to the chromatic number of the graph (see therein). Famously,
in the class of planar graphs graph coloring requires just $4$ colors
and polynomial time \cite{Robertson96} . The number of the evaders
needed for this reduction is $2$ because $2=\log_{2}4$ (see therein).

It would be sufficient to consider node interdiction since edge interdiction
is computationally equivalent to it:

\begin{lem}
The UME problem on edges is polynomially equivalent in complexity
to UME on nodes.
\end{lem}
\begin{proof}
The idea is standard: split each edge to create UME on nodes; to create
UME on edges break each node into two nodes connected by an edge.
See \cite{Gutfraind09unreactive} for details. 
\end{proof}
\begin{thm}
The class of UME problems on nodes with $2$ evaders is NP-hard.
\end{thm}
\emph{Proof: }Given an instance of the Planar Vertex Cover problem
$G'(V',E')$ with budget $B'$ construct an instance of UME node interdiction
on a derived graph $G(V,E)$ as follows in steps 1-3.\\
\\
Step 1: Graph Coloring. Run the algorithm \cite{Robertson96} on $G'$
and compute the color assignment: $f:V'\to\left\{ white,red,green,black\right\} $
(abbreviated $\left\{ w,r,g,b\right\} $).  

\hfill{}\\
Step 2: Construction of the UME graph. Assemble $G(V,E)$ as follows
(Fig.~\ref{fig:reduction1}):\\
(a) The nodes are copied from $V'$ and a special {}``target'' node
$t$ is added: $V=V'\cup\left\{ t\right\} $\\
(b) Include the original edges $E'$ and for all $u\in V'$ (non-singletons)
add an edge $(u,t)$ to $t$: $E=E'\cup\left\{ (u,t)\left|u\in V'\mbox{ and degree}(u)>0\right.\right\} $\\
(c) Define $\mathbf{d}$: All nodes $u\in V'$ can be completely interdicted:
$d_{uv}=1$ $\forall u,v$. \\

\begin{figure}
\begin{centering}
\includegraphics[width=1\columnwidth]{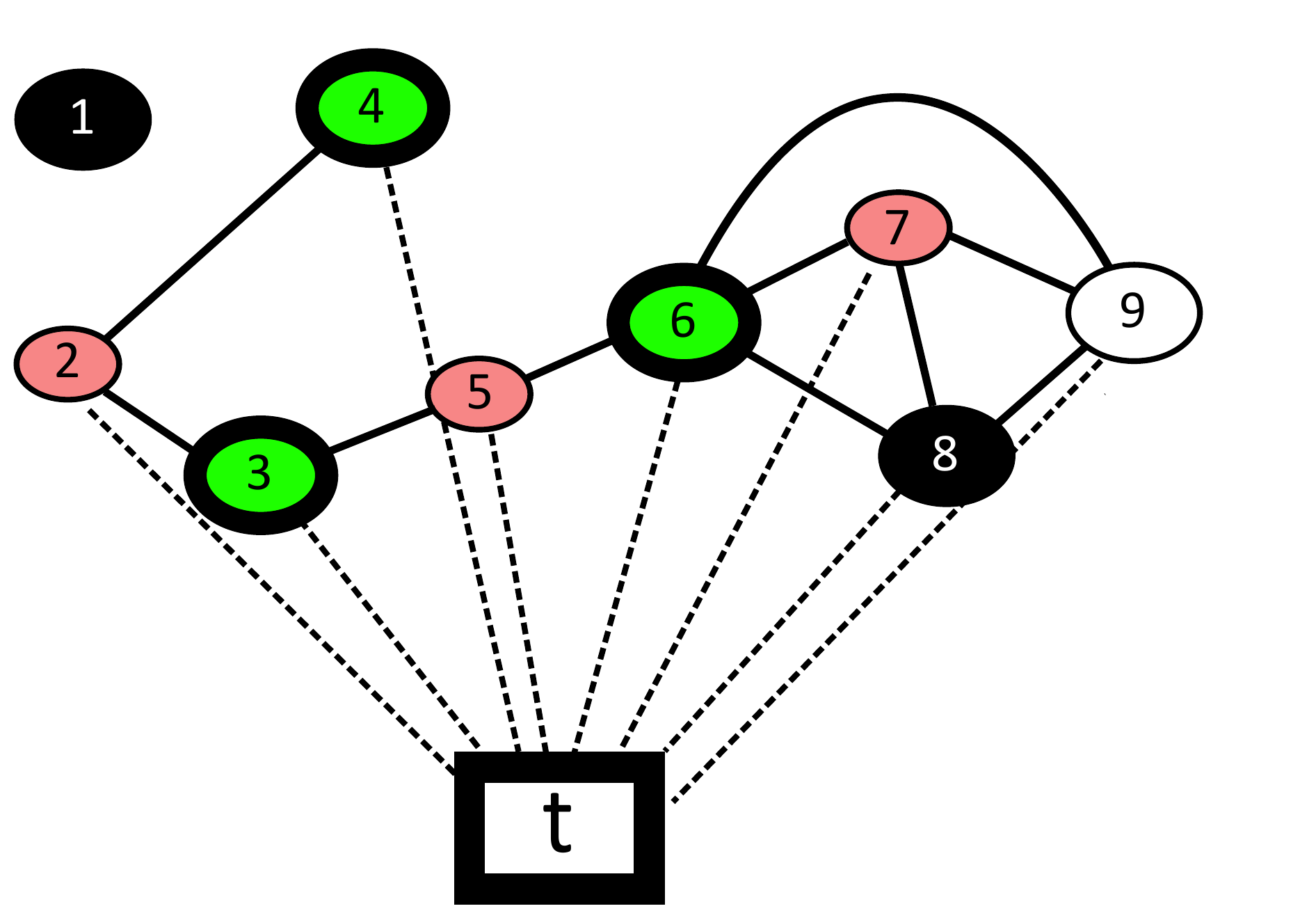}
\par\end{centering}

\caption{\label{fig:reduction1}The graphs $G'$ and $G$. The original Set
Cover instance is on $G'$ (drawn with elliptical nodes and solid
edges). $G$ is created by adding node $t$ (rectangle) and the edges
to $t$ (dashed lines).}

\end{figure}

\hfill{}\\
Step 3: Construction of the source distributions and transition matrices.
The two evaders would follow 3-node paths: from some source node through
a {}``penultimate node'' to the node $t$, as follows. \\
Define two sets of {}``source nodes'' $S_{1}$ and $S_{2}$ by including
all the non-singleton nodes with colors $\left\{ w,r\right\} $ in
$S_{1}$ and all non-singletons with colors $\left\{ w,g\right\} $
in $S_{2}$:

\begin{lyxlist}{00.00.0000}
\item [{$S_{1}=\left\{ u\left|u\in V'\smallsetminus\left\{ t\right\} \mbox{ and degree}(u)>0\mbox{ and }f(u)\in\left\{ w,r\right\} \right.\right\} $}] and 
\item [{$S_{2}=\left\{ u\left|u\in V'\smallsetminus\left\{ t\right\} \mbox{ and degree}(u)>0\mbox{ and }f(u)\in\left\{ w,g\right\} \right.\right\} $.}]~
\end{lyxlist}
Define two sets of {}``penultimate nodes'' $P_{1}$ and $P_{2}$
over all the non-singleton nodes with colors $\left\{ g,b\right\} $
in $P_{1}$ and all non-singletons with colors $\left\{ r,b\right\} $
in $P_{2}$:

\begin{lyxlist}{00.00.0000}
\item [{$P_{1}=\left\{ u\left|u\in V'\smallsetminus\left\{ t\right\} \mbox{ and degree}(u)>0\mbox{ and }f(u)\in\left\{ g,b\right\} \right.\right\} $}] and
\item [{$P_{2}=\left\{ u\left|u\in V'\smallsetminus\left\{ t\right\} \mbox{ and degree}(u)>0\mbox{ and }f(u)\in\left\{ r,b\right\} \right.\right\} $.}]~
\end{lyxlist}
Finally, introduce evaders $\left\{ 1,2\right\} $. For each evader
$i$, let $\mathbf{a^{\mathbf{(i)}}}$be uniformly distributed over
all nodes of class $S_{i}$ and define $\mathbf{M}^{(i)}$ so the
evader follows the 3-node path discussed earlier:

\begin{enumerate}
\item $M_{uv}^{(i)}=0$ if $u\notin S_{i}$ or $u=t$ or $v\notin P_{i}$
\item $M_{uv}^{(i)}=\frac{1}{z_{u}}$ if $u\in S_{i}$ and $v\in P_{i}$
where $z_{u}=\left\Vert \left\{ v\left|v\in P_{i}\mbox{ such that }(u,v)\in E\right.\right\} \right\Vert $
\item $M_{uT}^{(i)}=1$ if $u\in P_{i}\,.$
\end{enumerate}
An illustration of the evader motion is found in Fig.~\ref{fig:reduction2}.
In the pathological case where all nodes in $G'$ are singletons,
arbitrarily choose any node $u\neq t$ and for $i\in\left\{ 1,2\right\} $
let $\mathbf{a^{\mathbf{(i)}}}=\delta_{uv}$ with $\mathbf{M^{\mathbf{(i)}}}=\mathbf{0}$. 

\begin{lyxlist}{00.00.0000}
\item [{Observation}] 1: each of the non-singleton nodes belongs to one
of four disjoint set intersections, corresponding to the four colors
$\left\{ w,g,r,b\right\} $: $w\leftrightarrow S_{1}\cap S_{2}$,
$g\leftrightarrow S_{2}\cap P_{1}$, $r\leftrightarrow S_{1}\cap P_{2}$
and $b\leftrightarrow P_{1}\cap P_{2}$. These can be viewed as the
four bit strings of length 2: $00$, $01$, $10$ and $11$ (hence
$\log_{2}4$ evaders).
\item [{Observation}] 2: no source node coincides with a corresponding
penultimate node: $P_{1}\cap S_{1}=\emptyset$ and $P_{2}\cap S_{2}=\emptyset$.
Thus a direct jump from $S_{i}$ to $t$ has probability $=0$.
\item [{Observation}] 3: node $t$ could be pruned from any interdiction
set without changing the expected interdiction probability because
interdiction only affects outgoing evaders and node $t$ has none.
\end{lyxlist}
Notice also that the definitions of $\mathbf{a^{\mathbf{(i)}}}$ and
$\mathbf{M^{\mathbf{(i)}}}$ do not guarantee that there is a path
from every node of type $S_{i}$ to node $t$ but still meet the requirement
that each edge would be traversed by at least one evader. An extreme
example is a star graph such that $S_{1}\cap S_{2}$ nodes surround
a central node of type $S_{1}\cap P_{2}$: the edges are traversed
by evader $2$ and evader $1$ cannot reach node $t$.\\
\begin{figure}
\begin{centering}
\includegraphics[width=1\columnwidth]{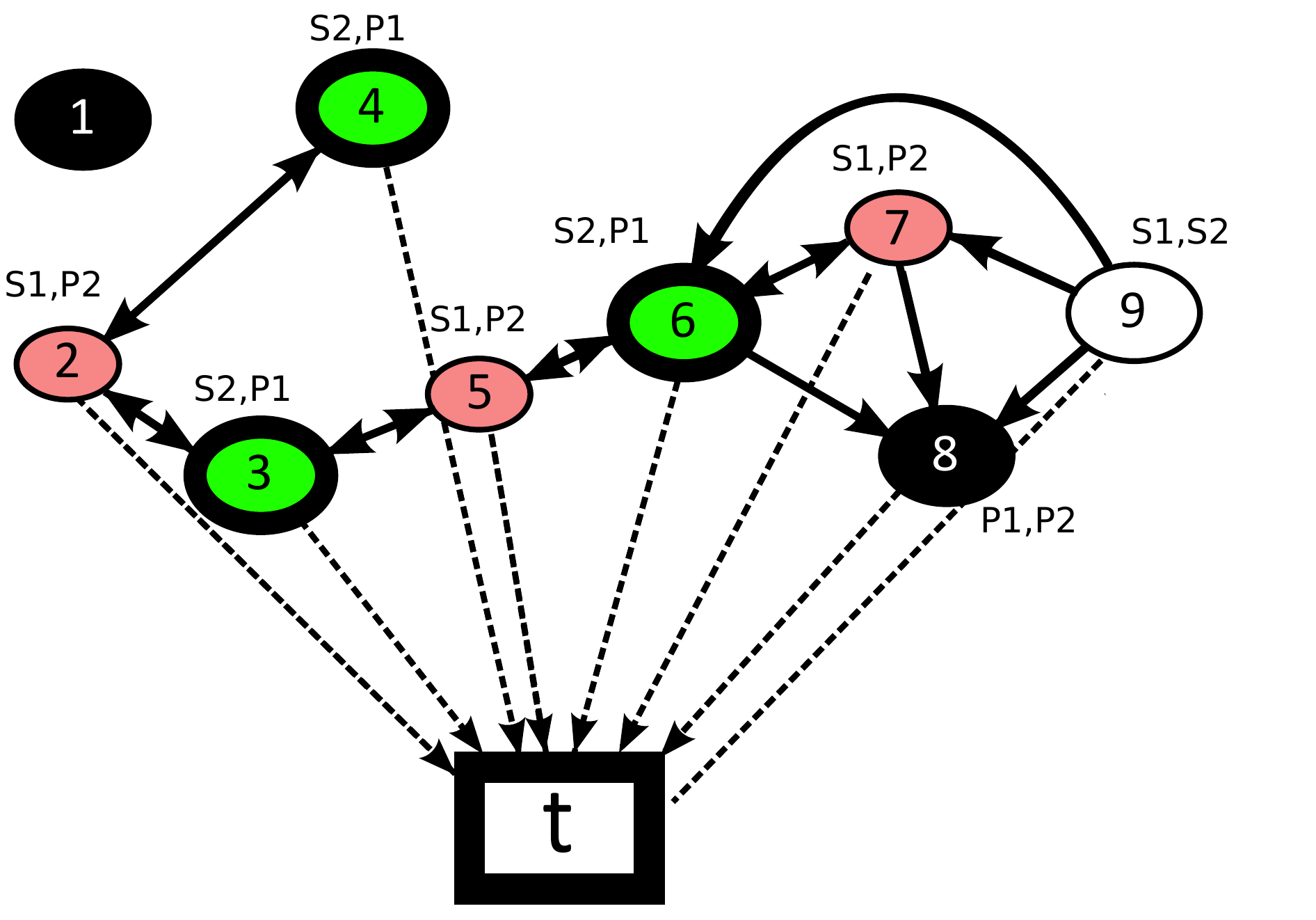}
\par\end{centering}

\caption{\label{fig:reduction2}The graph $G$ showing the evaders and classes
of the non-singleton nodes. \emph{White} indicates class $S_{1}\cap S_{2}$,
\emph{green} (large ellipses) indicates class $S_{2}\cap P_{1}$,
\emph{red} (small ellipses) indicates class $S_{1}\cap P_{2}$ and
\emph{black} indicates class $P_{1}\cap P_{2}$. Evader motion is
indicated by arrows. For example, the bi-directional arrow between
nodes $3$ and $5$ indicates that it is passed in both direction
by evaders: with $Pr>0$ evader $1$ moves along $5\to3\to t$ and
evader $2$ moves along $3\to5\to t$. }

\end{figure}

Define the UME decision problem: Is it possible to find an interdiction
set $Q$ of size at most $B$ so that expected interdiction probability
$\left\langle J\right\rangle =1$?

\begin{claim*}
The UME decision problem with budget $B$ set to $B'$ is a {}``YES''
instance iff a $B'$-cover exists for the graph $G'$.
\end{claim*}
\emph{Justification}: The pathological case where all nodes are singletons
is a UME {}``YES'' instance for any $B\geq0$ since the evader cannot
reach the target and it is also a Planar Vertex Cover {}``YES''
instance ($B'\geq0$) since no edges exist. 

Suppose now that a non-pathological UME instance is a {}``YES''
instance. Since adjacent nodes in $G'$ have different colors, \emph{Observation
1} implies that any two adjacent nodes $u,v\in V\smallsetminus\left\{ t\right\} $
must be different by at least one bit. Thus $\exists$evader $i$
such that one of $\left\{ u,v\right\} $ is a source node ($i^{th}$
bit $=0$) while the other is a penultimate node ($i^{th}$ bit $=1$).
The definitions of $\mathbf{a^{\mathbf{(i)}}}$ and $\mathbf{M^{\mathbf{(i)}}}$
imply that evader $i$ traverses through $(u,v)$ with $Pr>0$. Since
this is a {}``YES'' instance with $\left\langle J(Q)\right\rangle =1$,
the interdiction set $Q$ must contain at least one of the endpoints
$\left\{ u,v\right\} $ (whether or not node $t\in Q$, by \emph{Observation
3}). Therefore the set $Q$ is a cover for graph $G'$.

Conversely, if the Planar Vertex Cover decision problem is a {}``YES''
instance then there exists a vertex cover set $C$. From \emph{Observation
2} and the definition of $\mathbf{a^{\mathbf{(i)}}}$ it follows that
with $Pr=1$ the evader passes on his way to the target through the
set of edges in the original graph: $E'=\left\{ (u,v)\left|u,v\in V'\smallsetminus\left\{ t\right\} \mbox{ and }u\neq t\neq v\right.\right\} $.
Therefore make $Q=C$ and get that $\forall i$, evader $i$ will
be interdicted with expected probability $=1$. This a UME {}``YES''
instance.

\section{Discussion}

The proof in this paper shows that UME is NP-hard even under fairly
restrictive conditions: (1) only $2$ evaders are needed, (2) the
interdiction efficiencies $\mathbf{d}$ are everywhere $=1$, (3)
the graph is unweighted and undirected, and (4) the evader has the
non-retreating property \cite{Gutfraind08markovian}. It remains to
determine whether the result could be improved to the case of $1$
evader or to find a polynomial-time solution.  It is interesting
that the color-based technique introduced in the proof could be used
to solve other vertex cover problems, such as vertex cover for a bipartite
graph ($2$-colorable, so only $1$ evader is needed). Yet, on bipartite
graphs vertex cover can be solved in polynomial time \cite{Cormen01}.

\begin{acknowledgement*}
AG would like to thank Robert Kleinberg for fascinating lectures on
complexity.
\end{acknowledgement*}
\bibliographystyle{plain}
\bibliography{interdict,interdict_complexity}

\begin{thebibliography}{10}

\bibitem{Bar-noy-1995-complexity}
A.~Bar-Noy, S.~Khuller, and B.~Schieber.
\newblock The complexity of finding most vital arcs and nodes.
\newblock Technical report, University of Maryland, College Park, MD, USA,
  1995.

\bibitem{Berman95}
Oded Berman, Dmitry Krass, and Chen~Wei Xu.
\newblock Locating flow-intercepting facilities: New approaches and results.
\newblock {\em Annals of Operations Research}, 60:121--143, 1995.

\bibitem{Boros06-inapproximability}
E.~Boros, K.~Borys, and V.~Gurevich.
\newblock Inapproximability bounds for shortest-path network interdiction
  problems.
\newblock Technical report, Rutgers University, Piscataway, NJ, USA, 2006.

\bibitem{Corley-1982-most}
H.~W. Corley and D.~Y. Sha.
\newblock Most vital links and nodes in weighted networks.
\newblock {\em Oper. Res. Lett.}, 1(4):157 -- 160, Sep 1982.

\bibitem{Cormen01}
Thomas~H. Cormen, Charles~E. Leiserson, Ronald~L. Rivest, and Clifford Stein.
\newblock {\em Introduction to Algorithms}.
\newblock {MIT} Press, Cambridge, MA, USA, 2nd edition, 2001.

\bibitem{Garey74}
M.~R. Garey, D.~S. Johnson, and L.~Stockmeyer.
\newblock Some simplified np-complete problems.
\newblock In {\em STOC '74: Proceedings of the sixth annual ACM symposium on
  Theory of computing}, pages 47--63, New York, NY, USA, 1974. ACM.

\bibitem{Gutfraind08markovian}
Alexander Gutfraind, Aric Hagberg, David Izraelevitz, and Feng Pan.
\newblock Interdicting a {M}arkovian evader.
\newblock Preprint, 2009.

\bibitem{Gutfraind09unreactive}
Alexander Gutfraind, Aric~A. Hagberg, and Feng Pan.
\newblock Optimal interdiction of unreactive {M}arkovian evaders.
\newblock In John Hooker and Willem-Jan van Hoeve, editors, {\em {CPAIOR
  2009}}, Lecture Notes in Computer Science. Springer, May 2009.
\newblock http://arxiv.org/abs/0903.0173.

\bibitem{Pan-2003-models}
F.~Pan, W.~Charlton, and D.~P. Morton.
\newblock Interdicting smuggled nuclear material.
\newblock In D.L. Woodruff, editor, {\em Network Interdiction and Stochastic
  Integer Programming}, pages 1--19. Kluwer Academic Publishers, Boston, 2003.

\bibitem{Pourbohloul05}
B.~Pourbohloul, L.A. Meyers, D.M. Skowronski, M.~Krajden, D.M. Patrick, and
  R.C. Brunham.
\newblock Modeling control strategies of respiratory pathogens.
\newblock {\em Emerg. Infect. Dis.}, 11(8):1246--56, 2005.

\bibitem{Robertson96}
Neil Robertson, Daniel~P. Sanders, Paul Seymour, and Robin Thomas.
\newblock Efficiently four-coloring planar graphs.
\newblock In {\em STOC '96: Proceedings of the twenty-eighth annual ACM
  symposium on Theory of computing}, pages 571--575, New York, NY, USA, 1996.
  ACM.

\end{thebibliography}
{\footnotesize This article is released under Los Alamos National
Laboratory LA-UR-09-07611}
\end{document}